\documentclass{article}

\usepackage{styles/VABcommands}
\usepackage{styles/VABenvironments}
\usepackage[nobackref]{styles/VABcitations}
\usepackage{cleveref}

\DeclareMathOperator{\Sym}{S}

\title{On power sum kernels on symmetric groups}

\author{Iskander Azangulov$^1$ \and Viacheslav Borovitskiy$^2$ \and Andrei Smolensky$^1$}
\date{\itshape\small$^1$St. Petersburg State University, $^2$ETH Zürich}

\bibliography{references.bib}

\begin{document}

\maketitle

\begin{abstract}
In this note, we introduce a family of ``power sum'' kernels and the corresponding Gaussian processes on symmetric groups $\Sym_n$.
Such processes are bi-invariant: the action of $\Sym_n$ on itself from both sides does not change their finite-dimensional distributions.
We show that the values of power sum kernels can be efficiently calculated, and we also propose a method enabling approximate sampling of the corresponding Gaussian processes with polynomial computational complexity.
By doing this we provide the tools that are required to use the introduced family of kernels and the respective processes for statistical modeling and machine learning.
\end{abstract}
\section{Introduction}

Gaussian processes on various non-Euclidean domains are of great interest both from the point of view of pure mathematics \cite{taylor2003,istas2005,yaglom1961, cohen2012,cohen2013}, and in terms of various applications where they can be used as statistical models \cite{jaquier2022,hutchinson2021,coveney2020}.

Important are the cases when Gaussian processes and their kernels respect the geometric structure of their domain $X$.
Usually this means \emph{stationarity} in a generalized sense, that is, the invariance of all finite-dimensional distributions to the symmetry group of the space $X$.
For example, if $X$ is a Riemannian manifold, with respect to the isometry group, or, if $X$ is a group, with respect to the action of $X$ on itself.
In this note, we will be interested in the latter case: we will consider symmetric groups $X = \Sym_n$ and \emph{bi-invariant} Gaussian processes with their respective kernels, i.e. invariant under the action of $\Sym_n$ on itself by left and right translations \cite{yaglom1961}.
This, under the assumption that the Gaussian process is centered, is expressed by the relation
\[
k(\sigma g, \sigma h)
=
k(g \sigma, h \sigma)
=
k(g, h)
\]
for its kernel $k: \Sym_n \x \Sym_n \to \R$ and elements $g, h, \sigma \in \Sym_n$.

We define the family of \emph{power sum} kernels generating bi-invariant Gaussian processes on symmetric groups $\Sym_n$.
As far as the authors' knowledge goes, these were not considered before in the literature.
These kernels are defined in terms of the cycle type of the quotient of two permutations, and due to classical relations in representation theory are automatically non-negative definite.
Moreover, it turns out that these kernels can be computed efficiently and it is possible to sample the corresponding Gaussian processes efficiently, which makes these objectives attractive for statistical modeling of functions on $\Sym_n$.

For comparison, the typical covariances considered in the literature are defined as decreasing functions of some distance \cite{bachoc2020}, with different possible definitions of distance (see \cite[Chapter~6B]{diaconis1988}) yielding different covariances functions.
In all these cases, it is required to separately prove their non-negative definiteness.
The covariances we define are not, in general, functions of distance in any standard sense, but they remain monotonic along shortest paths in the Cayley graph of the symmetric group.

\section{Power sum kernels}
We call a finite sequence of natural numbers $\lambda = (\lambda_1,\lambda_2,\ldots,\lambda_\ell)$ with $\lambda_1 \geq \lambda_2 \geq \ldots \geq \lambda_\ell$ and with $\lambda_1 + \ldots + \lambda_\ell = n$ a \emph{partition} of $n$ of \emph{height} $\ell$.

Let $d \in \Z_+$, $m \in \N$, and let $p_d: \C^m \to \C$ denote the Newton's power sum, that is, a symmetric polynomial of the form
\begin{equation}
p_d(z_1, \ldots, z_m)
=
z_1^d + z_2^d + \ldots + z_m^d.
\end{equation}
For a partition $\lambda = (\lambda_1,\lambda_2,\ldots,\lambda_\ell)$ define $p_\lambda = \prod_{j=1}^\ell p_{\lambda_j}$.
Finally, for a permutation $g \in \Sym_n$ define $\mu(g) =(\mu_1, \mu_2, \ldots, \mu_\ell)$, where $\ell$ is the total number of cycles in permutation $g$, and $\mu_j$ is the length of the $j$-th largest cycle.

We will define a new family of stationary kernels on $\Sym_n$, which we will call \emph{power sum kernels} in what follows.
\begin{definition}
For a vector $\v{z}=(z_1,\ldots, z_m)\in \R^m$, where $z_j \geq 0$, set
\begin{equation}
k_{\v{z}}(g,h) = p_{\mu(gh^{-1})}(\v{z}).
\end{equation}
\end{definition}

\begin{theorem} \label{thm:main_thm}
Function $k_{\v{z}}: \Sym_n \x \Sym_n \to \R$ is positive definite and bi-invariant, i.e. $k_{\v{z}}(\sigma g, \sigma h) = k_{\v{z}}(g \sigma, h \sigma) = k_{\v{z}}(g, h)$ for all $\sigma \in \Sym_n$.
Moreover, if the permutation $gh^{-1}$ has cycle type $(c_1, \ldots, c_n)$, i.e. the permutation $gh^{-1}$ has $c_j$ cycles of length $j$ with $j = 1, \ldots, n$, then
\begin{equation} \label{eqn:ker_komp}
k_{\v{z}}(g,h) = \prod_{j=1}^n p_j(z)^{c_j}.
\end{equation}
\end{theorem}

The proof of \Cref{thm:main_thm} will be given in~\Cref{sec:main_thm_proof}.
For now we will focus on simple properties and examples of kernels from this family.

\subsection{Properties and examples of power sum kernels}

First, note that~\Cref{eqn:ker_komp} does obviously make it possible to compute the values of power kernels $k_{\v{z}}(g,h)$ efficiently.

Since $k_{\v{z}}(g,h)$, as a function of $z$, is a homogeneous polynomial of degree~$n$, we have $k_{r\v{z}}(g,h) = r^n k_{\v{z}}(g,h)$.
In case $\sum_{j=1}^m z_j = 1$ one gets that $k_{\v{z}}(g,g)=(\sum_{j=1}^m z_j)^n = 1$, which implies that the variance $k_{\v{z}}(g, g)$ equals $|\v{z}|_{1}^n$.

In what follows, without loss of generality, we assume that $k_{\v{z}}(g, g) \equiv 1$.
With this normalization, it can be seen that the constructed covariance multiplicatively penalizes each cycle of length $j$, multiplying the result by $z_1^j + \ldots + z_m^j$.

Consider the Cayley graph of the group $\Sym_n$ with respect to the set of all transpositions (permutations that swap only two elements).
This is a graph $G = (V, E)$ whose vertex set $V$ coincides with $\Sym_n$, and the set of edges $E$ consists of those pairs $(g, h)$ for which the permutation $g h^{-1}$ is a transposition.

The distance between permutations $g$ and $h$ in this graph (commonly called the Cayley distance) can be computed as
\begin{equation} \label{eqn:caley_dist}
d_C(g,h) = n-|\operatorname{cycles}(g h^{-1})|,
\end{equation}
where $\operatorname{cycles}(g)$ denotes the set of cycles of $g$, including cycles of length $1$.
Note that $|\operatorname{cycles}(g h^{-1})| = \ell$, where $\ell$ is the total number of cycles of $gh^{-1}$ from the definition of $\mu(g h^{-1})$.

Let us introduce the following order relation on $\Sym_n$: we say that $h \prec g$ if there is a shortest path in the Cayley graph from vertex $e$ (the identity permutation) to vertex $g$ passing through vertex $h$.
\begin{proposition}
If $h \prec g$ and $\v{z} \in \R^m$ for $m > 1$ then $k_{\v{z}}(h,e) \geq k_{\v{z}}(g, e)$.
If $z_j > 0$ for at least two distinct indices $j$, then $k_{\v{z}}(h,e) > k_{\v{z}}(g, e)$.
\end{proposition}
\begin{proof}
For each vertex of the Cayley graph along the shortest path from~$g$ to~$e$, the number of cycles increases by $1$ by formula~\eqref{eqn:caley_dist}.
Hence, it suffices to prove that $k_{\v{z}}(g,e) \leq k_{\v{z}}(g\sigma, e)$, where $\sigma$ is a permutation that increases the number of cycles by $1$.
Let the permutation split some cycle of length $a$ into two cycles of lengths $b$ and $c$ (so that $a = b+c$). Then
\begin{equation}
k_{\v{z}}(g, e)
=
k_{\v{z}}(g \sigma, e)
\frac{(z_1^a + z_2^ a + \ldots + z_m^a)}
{(z_1^b + z_2^b + \ldots + z_m^b)(z_1^c + z_2^c + \ldots + z_m^c)}
\leq
k_{\v{z}}(g \sigma, e),
\end{equation}
since expanding the product in the denominator produces the sum of the numerator and terms of the form $z_j^b z_l^c \geq 0$.
If there are distinct indices $j, l$ for which $z_j > 0$ and $z_l > 0$, then $z_j^b z_l^c > 0$ and the strict inequality holds.
\end{proof}

Substituting different values of $z$, we can see that the family of power kernels includes the following Gaussian processes as ``extreme'' cases.
\begin{itemize}
\item For $\v{z}=(1) \in \R^1$ one gets 
\begin{equation}
k_{\v{z}}(g, h) =\prod_{j=1}^{n} 1^{c_j} \equiv 1.
\end{equation}
Therefore, if $f$ is a centered Gaussian process with covariance $k_{\v{z}}$, then $f(g) = f(e)$, $g \in \Sym_n$, where $f(e) \sim \f{N}(0,1) $.
\item For $\v{z}=(1, \ldots, 1) \in \R^m$ one gets
\begin{equation} \label{eqn:zonesvector}
k_{\v{z}}(g,h)
=
\prod_{j=1}^\ell
\Big(\sum_{l=1}^m 1^{\mu_j(gh^{-1})}\Big) = m^\ell
\end{equation}
where, by the definition of $\mu$, the number $\ell$ is the total number of cycles of the permutation $g h^{-1}$. 
In particular, in this case, the covariance is a function of the Cayley distance.
Note that the frequently used \emph{Cayley distance kernels}, given by $\exp(-\beta\cdot d_C(g,h))$, are positive semidefinite under the sufficient condition $\exp(\beta)\in\{ 1,\ldots,n-1 \} \cup [n,\infty)$ \cite{corfield2021}, and~\Cref{eqn:zonesvector} corresponds precisely to the Cayley distance kernel with $\log$-integer ``inverse temperature'' $\beta$.
\item For $\v{z}=(1/m, \ldots, 1/m) \in \R^m$, using homogeneity, it is easy to see that for any permutation other than the identity the following inequality holds:
\begin{equation}
k_{\v{z}}(g,e) = m^{k-n} \le m^{-1}.
\end{equation}
This means that for $m\to\infty$ the covariance $k_{\v{z}}$ converges to delta function, and the process itself converges to the Gaussian white noise on $\Sym_n$.
\end{itemize}

Some of the power sum kernels are illustrate on~\Cref{fig:s6-kernels}.

\begin{figure}
    \centering
    \includegraphics[scale=0.5]{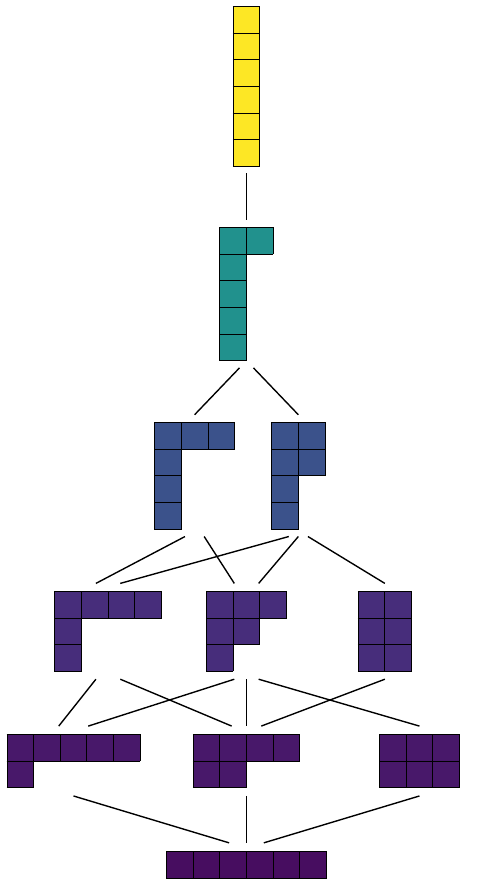}%
    \hfill
    \includegraphics[scale=0.5]{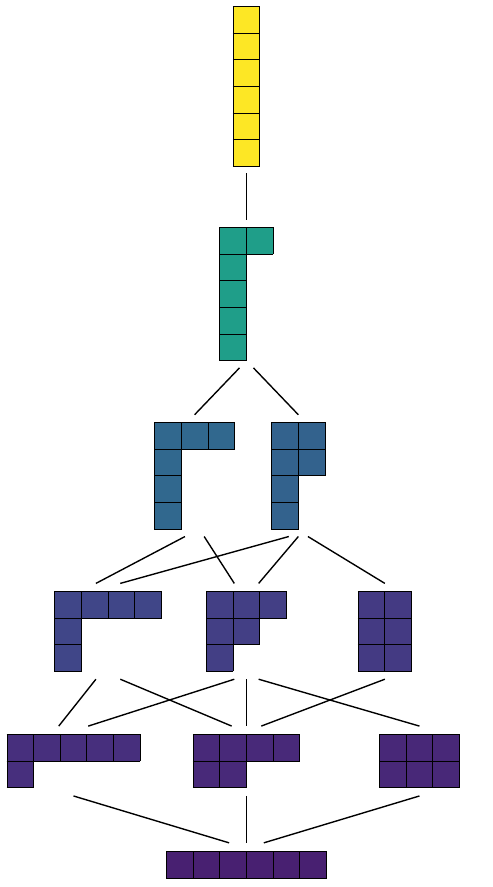}%
    \hfill
    \includegraphics[scale=0.5]{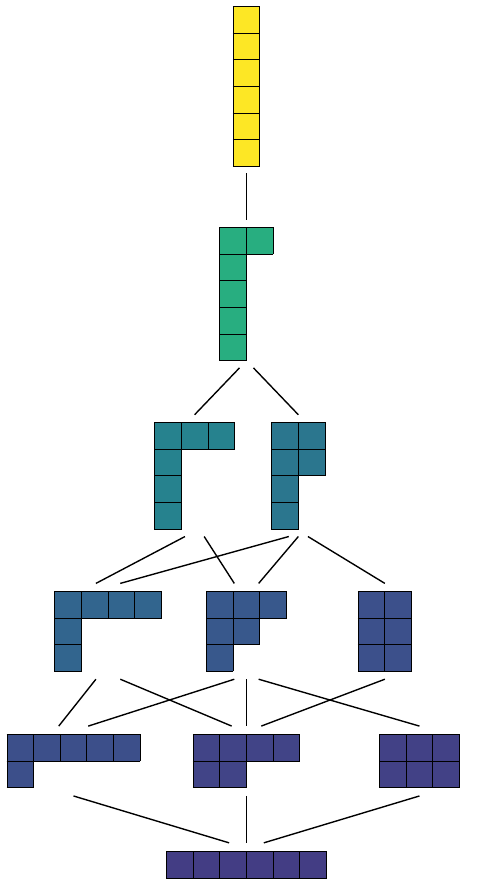}%
    \hfill
    \includegraphics[scale=0.5]{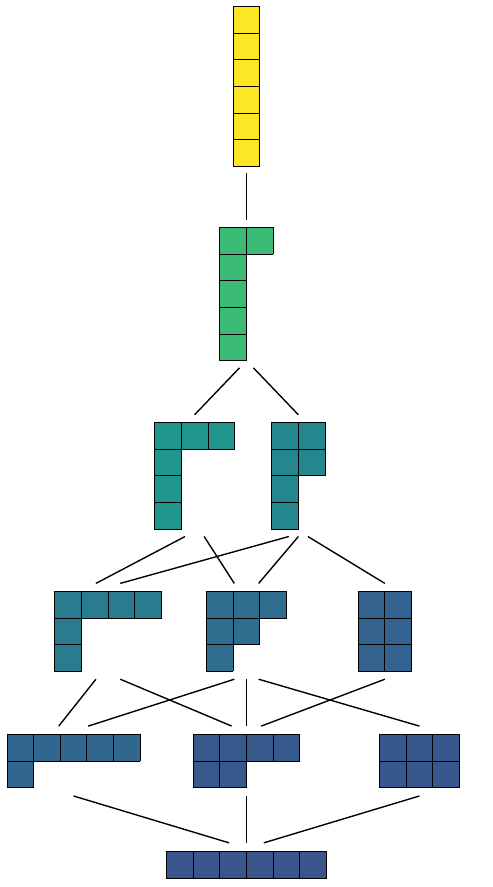}
    \caption{The values of power-sum kernels on $\Sym_6$, corresponding to $\v{z} = (1/2,1/2)$, $(2/3,1/3)$, $(3/4,1/4)$ and $(4/5,1/5)$. The underlying graph is the quotient of the Cayley graph of $\Sym_6$ with vertices partitioned into conjugacy classes, given by Young diagrams.}
    \label{fig:s6-kernels}
\end{figure}

\subsection{Proof of the theorem~\ref{thm:main_thm}} \label{sec:main_thm_proof}

Yaglom's theory \cite{yaglom1961} implies that bi-invariant kernels on compact groups can be described in terms of representation theory.
In particular, any bi-invariant kernel $k: G \x G \to \R$ on a compact group $G$, whose identity element we denote by $e_G$, has the form
\begin{equation}\label{eqn:stationary_kernels}
k(g, h)
=
k(gh^{-1}, e_G)
=
k(e_G, g^{-1}h)
=
\sum_\lambda a_\lambda\chi_\lambda(gh^{-1}),
\end{equation}
where $a_\lambda \geq 0$ and $\sum_\lambda a_\lambda d_\lambda < \infty$, the summation is over all irreducible unitary representations $\pi_\lambda$ of the group $G$, $\chi_\lambda\colon G \to \C$ denotes the character of the representation $\pi_\lambda$, and $d_\lambda$ denotes its dimension.
Moreover, any function of the form~\eqref{eqn:stationary_kernels} is bi-invariant and positive semidefinite (thus a well-defined kernel).

The group $\Sym_n$ is finite, and therefore compact.
The number of its irreducible unitary representations is finite.
It is known that irreducible representations of the group $\Sym_n$ can be indexed by Young diagrams of size $n$.
A Young diagram is the representation of a partition $\lambda$ in form of the table with rows of size $\lambda_1, \lambda_2, \ldots, \lambda_\ell$.
In this case, the number $\ell$ is called the length of the Young diagram, and $n =\sum_{j=1}^\ell \lambda_j$ is called its weight.
It known that the characters $\chi_{\lambda}$ in the case of $\Sym_n$ are real-valued and even integer \cite[Section~13.1, Corollary~1]{serre1977}.

Young diagrams are also closely related to the unitary group $\f{U}(m)$: diagrams of length at most $m$ enumerate the irreducible representations of $\f{U}(m)$.
Moreover, there is a deep connection between the representations of $\Sym_n$ and $\f{U}(m)$, expressed through the Schur--Weyl duality.
The character $s_{\lambda}$ of the irreducible representation of the group $\f{U}(m)$ corresponding to the Young diagram $\lambda$ can be expressed as the Schur polynomial in the eigenvalues $z_1, \ldots, z_m$ of the unitary matrix $\m{Z}$.
Using Weyl's formula \cite[\S 7.15]{stanley1997} it can be expressed by:
\begin{equation}
\label{eqn:schur_det}
s_{\lambda}(\m{Z})
=
s_{\lambda}(z_1, \ldots, z_m) =
\frac{\det \del{z_i^{\lambda_j + m - j}}_{i, j = 1}^m}
{\det \del{z_i^{m - j}}_{i, j = 1}^m},
\end{equation}
where for $j$ greater than the length of $\lambda$ it is assumed that $\lambda_j = 0$, and the denominator divides the numerator, implying that $s_{\lambda}$ is indeed a polynomial in the variables $z_1, \ldots, z_m$.
Note that representations of the group $\f{U}(m)$ can be extended to representations of the ambient general linear group $\f{GL}(m,\C)$, whose character is also given by~\Cref{eqn:schur_det}.

There is another representation for $s_{\lambda}$, given in terms of semistandard Young tableaux, which we will now define.
Let $\lambda$ be a partition.
A semistandard tableau of shape $\lambda$ is a Young diagram corresponding to the partition $\lambda$, filled with natural numbers in such a way that values do not decrease in rows and increase in columns.
Polynomials $s_{\lambda}$ can be described as follows~\cite[\S 7.10]{stanley1997}:
\begin{equation}
\label{eqn:schur_tabluax}
s_\lambda(\v{z})
=
\sum_T \v{z}^T
=
\sum_T \prod_{j=1}^m z_j^{t_j},
\end{equation}
where the summation is over all semistandard tables $T$ of shape $\lambda$ filled with numbers from $1$ to $m$, and $t_j$ denotes the number of cells in the table~$T$ containing the number $j$.

The set of polynomials $s_{\lambda}$, similarly to the set of polynomials $p_{\lambda}$, forms a basis of the space of symmetric polynomials, while the characters $\chi_{\lambda}$ of the group $\Sym_n$ define the elements of the basis change matrix.
More formally, the following holds.

\begin{theorem*}[Frobenius Formula, {\cite[\S4.1]{fulton1991}}]
\begin{equation}
\label{eqn:frobenius}
p_{\mu(g)}(z_1, z_2, \ldots, z_m)  = \sum \chi_\lambda(g) s_\lambda(z_1, z_2, \ldots, z_m),
\end{equation}
where the sum is over all Young diagrams of weight $n$ and length at most $m$, that is, partitions $\lambda$ of $n$ into at most $m$ terms.
\end{theorem*}

We are now ready to prove the theorem~\ref{thm:main_thm}.

Using \eqref{eqn:frobenius}, we write
\begin{equation}
\label{eqn:kernel_schur}
k_{\v{z}}(g,h)
=
p_{\mu(gh^{-1})}(\v{z})
=
\sum \chi_\lambda(gh^{-1}) s_\lambda(z_1, z_2, \ldots, z_m).
\end{equation}
Since $z_j \geq 0$ for $j = 1, \ldots, n$ by assumption, we see from the equality~\eqref{eqn:schur_tabluax} that $s_\lambda(z_1, z_2, \ldots, z_n) \geq 0$.
It follows that $k_{\v{z}}$ is of form~\eqref{eqn:stationary_kernels} and is therefore a bi-invariant covariance on $\Sym_n$.

Finally, by definition of $\mu(gh^{-1})$ we have $\abs[0]{\cbr[0]{\mu_j(gh^{-1}) = r}} = c_r$.
Taking this into account and grouping polynomials by their respective values of $\mu_j$, we get
\begin{equation}
k_{\v{z}}(g,h)
=
p_{\mu(gh^{-1})}(\v{z})
=
p_{\mu_j(gh^{-1})}(\v{z})
=
\prod_{r=1}^n p_r^{c_r}(\v{z})
.
\end{equation}

\section{Sampling power sum Gaussian processes}

Here we discuss ways of sampling Gaussian processes with power sum covariance functions, i.e. computational algorithms for obtaining their (approximate) trajectories.
Let $f$ be a centered Gaussian process on the group $\Sym_n$ with covariance $k\colon \Sym_n \x \Sym_n \to \R$.
Let $\v{f} =\del{f(g_1), \ldots, f(g_{n!})}^{\top}$, where $g_1, \ldots, g_{n!}$ are arbitrarily enumerated elements of $\Sym_n$.
Then $\v{f} \sim \f{N}(\v{0}, \m{K})$, where $\m{K}$ is the covariance matrix of the random vector~$\v{f}$.
Its elements are $\m{K}_{j r} = k(g_j, g_{r})$.
Sampling the process $f$ means sampling the vector $\v{f}$.
Let us write
\begin{align}
\label{eqn:sq_root_repr}
\v{f}
=
\m{K}^{1/2} \v{\varepsilon},
&&
\v{\varepsilon} \sim \f{N}(\v{0}, \m{I}),
\end{align}
where $\m{K}^{1/2}$ is a root of the matrix $\m{K}$ in the sense that $\del{\m{K}^{1/2}}^{\top} \m{K}^{1/2} = \m{K}$, and $\m{I}$ is the identity matrix of size $n!$.

The representation~\eqref{eqn:sq_root_repr} makes it possible to algorithmically sample $f$ by 1)~calculating the root $\m{K}^{1/2}$, 2)~sampling the standard normal vector $\v{\varepsilon} \sim \f{N}(\v{0}, \m{I})$, and 3)~computing the matrix-vector product $\m{K}^{1/2} \v{\varepsilon}$.
Step (2) has computational complexity $O(n!)$, step (3) has computational complexity $O((n!)^2)$, and step (1) has computational complexity $O((n!)^3)$.\footnote{The power $3$ can be reduced by using specialized algorithms such as Strassen algorithm. However, this exponent will always be no less than $2$ and such specialized algorithms are rarely used in practice because of the large constants hidden behind the $O$ notation.}
This ultimately makes this approach untenable for computing.
This is not surprising given its excessive generality and not even relying on bi-invariance in any way.
In what follows, we propose other techniques that leverage bi-invariance to allow, albeit approximately, sampling $f$ in a much more computationally efficient way.

The first such technique is a special case 
of the \emph{generalized random phase Fourier features} method proposed by \cite{azangulov2022} for modeling bi-invariant Gaussian processes on general compact groups.

Let $k\colon \Sym_n \x \Sym_n \to \R$ be a bi-invariant covariance.
Then it is of form~\eqref{eqn:stationary_kernels}, i.e. $k(g, h) = \sum_\lambda a_\lambda \chi_\lambda(g h^{-1})$.
Thus \cite{azangulov2022} gives
\begin{align}
f(g)
\approx
\sum_{r=1}^R
\frac{1}{\sqrt{L}}
\sum_{l=1}^L
w_l^{\lambda_r}
\chi_{\lambda_r}(g u_{l}),
&&
w_l^{\lambda_r} \stackrel{\text{i.i.d.}}{\sim} \f{N}\del{0,\frac{a_{\lambda_r}}{d_{\lambda_r}}},
\end{align}
where $u_{l} \stackrel{\text{i.i.d.}}{\sim} \f{U}(\Sym_n)$ with $\f{U}(\Sym_n)$ denoting the uniform distribution (Haar probability measure) on $\Sym_n$, $d_{\lambda}$ denotes the dimension of the representation~$\pi_{\lambda}$, and $\lambda_1, \ldots, \lambda_R$ correspond to the $R$ representations of the group $\Sym_n$ with the highest value $a_\lambda d_\lambda$.
To apply this approach, one has to be able to find $R$ representations with the largest value $a_\lambda d_\lambda$.
The number of representations of $\Sym_n$  grows as $\del{4n\sqrt{3}}^{-1} e^{\pi\sqrt{2n/3}}$ \cite{erdos1942}, thus it is not clear how to do this in a computationally tractable way.
Because of this we propose another technique for approximate modeling that does not require such a search.

It is similar to the technique from~\cite{azangulov2022} described above, but it is intended for bi-invariant processes on $\Sym_n$ specifically.
To do this, we need some preliminaries.

As before, by $s_\lambda(\m{X})$, where $\m{X} \in \f{GL}(m, \C)$ is a diagonalizable matrix, we denote $s_\lambda(x_1, x_2, \ldots, x_m)$, where $x_j$ are eigenvalues of the matrix $\m{X}$.
By the symbol $\f{N}_{\C^{m \x m}}$ we denote the ensemble of random matrices of size $m \x m$ with complex entries of the form $\f{N}(0, 1/2) + i \f{N}(0, 1/2)$, where all real and imaginary parts are jointly independent. 
Then the following is true~\cite{orlov2002, forrester2009}:
\begin{equation}
\label{eqn:schur_ginibre}
\begin{aligned}
\E_{\m{Z} \sim \f{N}_{\C^{m\times m}}}
s_\lambda(\m{A}\m{Z})
\overline{s_\lambda(\m{A}\m{Z})}
&=
\pi^{-m^2}
\int_{-\infty}^{\infty}
s_\lambda(\m{A}\m{Z})
\overline{s_\lambda(\m{A}\m{Z})}
e^{-\tr(\m{Z}\m{Z}^*)}
\d \m{Z}
\\
&=
\frac{n! \cdot s_\lambda(\m{A}\m{A}^*)}{d_\lambda}.
\end{aligned}
\end{equation}
Also for any character $\chi_\lambda$ of an irreducible representation of $\Sym_n$ we have \cite{azangulov2022}
\begin{equation}
\label{eqn:ch_average}
\frac{\chi_\lambda(gh^{-1})}{d_\lambda}
=
\int_{\Sym_n} \chi_\lambda(g u)\chi_\lambda(h u)
=
\E_{u \sim \f{U}(\Sym_n)} \chi_\lambda(g u)\chi_\lambda(h u).
\end{equation}
Rewriting the right hand side of the formula~\eqref{eqn:kernel_schur} using \eqref{eqn:schur_ginibre} with
\begin{equation}
\m{A}
=
\f{diag}(\sqrt{z_1},\sqrt{z_2},\ldots,\sqrt{z_m})
\end{equation}
and using~\Cref{eqn:ch_average} afterwards, we get that
\begin{equation}
\begin{aligned}
k_{\v{z}}(g h^{-1})
&=
\sum_{\lambda} \chi_\lambda(g h^{-1}) s_\lambda(z_1, z_2, \ldots, z_m)
\\
\label{eqn:sampling_before_rh}
&= 
\sum_\lambda
\frac{d_\lambda^2}{n!}
\E_{u \sim \f{U}(\Sym_n)}
\E_{\m{Z} \sim \f{N}_{\C^{m\times m}}}
\chi_\lambda(g u)
\chi_\lambda(h u)
s_\lambda(\m{A} \m{Z})
\overline{s_\lambda(\m{A} \m{Z})}.
\end{aligned}
\end{equation}

Finally, we need some properties of the Robinson---Schensted algorithm \cite{knuth1998}.
Given a permutation $g\in \Sym_n$, this algorithm constructs in a bijective way a pair of standard Young tableaux $(P(g),Q(g))$ of the same shape $\lambda(g)$.
We omit all the detail details concerning this mapping, but note only that there is an implementation with computational complexity $O(n^{3/2})$~\cite{tiskin2022}.
For $\lambda(g)$ from the Robinson---Schensted algorithm and for $g \sim \f{U}(\Sym_n)$ the equality $\P(\lambda(g) = \lambda) = d_\lambda^2 / n!$ holds.
Therefore, summing over $\lambda$ in the expression~\eqref{eqn:sampling_before_rh} can be rewritten so that
\begin{equation}
k_x(gh^{-1})
=
\E_{\substack{v, u \sim \f{U}(\Sym_n) \\ \m{Z} \sim \f{N}_{\C^{m\times m}}}}
\chi_{\lambda(v)}(gu)
\chi_{\lambda(v)}(hu)
s_{\lambda(v)}(\m{A}\m{Z})
\overline{s_{\lambda(v)}(\m{A}\m{Z})}.
\end{equation}
As a result, we obtain the desired approximation of the following form:
\begin{equation}
f(g)
\approx
\frac{1}{\sqrt{L}}
\sum_{j=1}^L
\del{w_{j, 1}
\operatorname{Re}s_{\lambda(v_j)}(\m{A} \m{Z}_j)+w_{j,2}\operatorname{Im}s_{\lambda(v_j)}(\m{A} \m{Z}_j)}
\chi_{\lambda(v_j)}(g u_j),
\end{equation}
with random coefficients
$
w_{j, 1}\stackrel{\text{i.i.d.}}{\sim} \f{N}(0,1)
$,
$
w_{j, 2} \stackrel{\text{i.i.d.}}{\sim} \f{N}(0,1)
$,
and indices
$
v_j \stackrel{\text{i.i.d.}}{\sim} U(\Sym_n)
$,
$
u_j \stackrel{\text{i.i.d.}}{\sim} U(\Sym_n)
$,
$
\m{Z}_j \stackrel{\text{i.i.d.}}{\sim} \f{N}_{\C^{m\times m}}
$.

\printbibliography

\end{document}